\documentclass[11pt,dvipsnames]{scrartcl}
\usepackage{titling}
\usepackage{amsfonts,amssymb,amsmath,amsthm,xspace}
\usepackage[sort,comma,numbers]{natbib}
\usepackage{enumitem}
\usepackage{tikz}
\usetikzlibrary{arrows}
\usepackage{mathtools}

\usepackage[font=small]{caption}
\usepackage{graphicx}
\usepackage{tikz}
\usetikzlibrary{calc}
\usetikzlibrary{positioning}
\pgfdeclarelayer{bg}    
\pgfsetlayers{bg,main}  

\usepackage[utf8]{inputenc}
\usepackage[ocgcolorlinks, linkcolor={blue}, citecolor={blue},
urlcolor={blue}]{hyperref}
\usepackage{cleveref}

	\DeclareMathOperator{\poly}{Poly}

\newtheorem{theorem}{Theorem}
\newtheorem{lemma}[theorem]{Lemma}

\newtheorem{observation}[theorem]{Observation}

\definecolor{sz}{rgb}{0.1,0.2,0.6}
\definecolor{blue}{rgb}{0.1,0.2,0.7}
\definecolor{brown}{rgb}{0.6,0.6,0.2}

\newcommand{\N}{\mathbb{N}}

\title{Fine-Grained Complexity of MIS on Regular Graphs}

\author{Saeed Akhoondian Amiri\thanks{Department of Computer Science, University of Cologne, Germany, amiri@cs.uni-koeln.de.}}

\begin{document}
	
	\maketitle

	\begin{abstract}
		\noindent 
		
		We show that there is no subexponential time algorithm for computing the exact solution of the maximum independent set problem in $d$-regular graphs, unless ETH fails. We expand our method to show that it helps to provide lowerbounds for other covering problems such as vertex cover and clique. We utilize the construction to show the NP-hardness of MIS on 5-regular planar graphs, closing the exact complexity status of the problem on regular planar graphs.
	\end{abstract}

	\section{Introduction}\label{sec:intro}
	
	In this paper, we investigate the hardness of the exact computation of the maximum independent set problem in graphs. Independent set problem is among fundamental problems in graph theory that asks for a set of vertices that are pairwise non-adjacent; we are interested in finding such a set of maximum size and this is the maximum independent set (MIS) problem. 
	MIS problem is hard to approximate within any constant factor in general graphs\cite{clique}, however, on bounded degree graphs, any maximal independent set, which can be computed in linear time, is a constant-factor approximation of the MIS. 
	
	Given the hardness of approximation of the problem, it is natural to ask for the efficiency of algorithms that run in superpolynomial time. The first step is to discover its exact complexity: is it possible to find an MIS in superpolynomial time but better than a trivial approach? If yes to what extent can we speedup such an algorithm?
	
	Many NP-hard problems are solvable exactly in time $O(2^n\poly(n))$. In particular for the independent set problem the trivial algorithm of testing all possible solutions yields such a running time.
	There are several improvements over the trivial upper bound both in general graphs and graphs of bounded degree~\cite{BourgeoisEPR12,XiaoN17,FominGK09,Robson86}; however all of them have a running time of the form $O(c^n)$ for a certain constant $c<2$. Hence, all of the mentioned algorithms are exponential to $n$. On the other hand, there are NP-hard problems that are subexponential time solvable; e.g.\ by exploiting bidimensionality theory, it is possible to solve several NP-hard problems in time $2^{O(\sqrt{n})}$ in excluded minor graphs~\cite{DemaineFHT04}. Additionally, MIS problem admits subexponential time algorithms on $P_t$-free graphs ~\cite{BacsoLMPTL19,Brause17}. Thus, both theoretically and practically it is interesting to discover NP-hard problems with subexponential time algorithms.
	
	One of the main tools designed to better understand the exact complexity of computational hard problems is the Exponential Time Hypothesis (ETH)~\cite{ImpagliazzoPZ01}. 
	Assuming ETH, there is no algorithm with running time $2^{o(n)}$ to solve the $3$-SAT problem, where $n$ stands for the number of variables in the formula. It has been proven that under the same assumption, several other major problems have no subexponential time algorithms. There are several other known problems that do not admit subexponential time algorithms under ETH, for instance, the $k$-SAT problem~\cite{PaturiPSZ05} and the maximum independent set problem on bounded degree graphs~\cite{JohnsonS99} are among such problems.
	
	Both of the above results, for $k$-SAT and MIS problem on bounded degree graphs, are showing the hardness of exact computation already among the sparse instances. To show the hardness in sparse instances, the main challenge is to transfer a general instance to an instance that is sparse or it has certain structural properties. If the transformation takes subexponential time then we can connect them to the existing known problems to show that such special cases do not admit subexponential time algorithm under the assumption of ETH. Such reductions are mostly known as sparsification lemmas.

	In this work, we continue a similar spirit by transforming from generic bounded degree graphs to $d$-regular graphs, which are quite restricted well-structured graphs. A graph $G$ is $d$-regular if degree of every vertex of $G$ is exactly $d$. We show that for every integer $d>2$ the MIS problem has no subexponential algorithm unless ETH fails.
	
	One of the related work is the result of Mohar~\cite{Mohar01}; he showed that the MIS problem and the vertex cover problem are NP-complete in $3$-regular planar graphs. His reduction in a sense is similar to the one of Johnson and Szegedy~\cite{JohnsonS99} in general graphs: they used, by now the standard, technique of replacing vertices of high degree by paths/cycles and then analyzing the connections. In particular, Johnson and Szegedy proved the following theorem.
	
	\begin{theorem}[Johnson and Szegedy~\cite{JohnsonS99}]\label{thm:ref8}
		Assuming ETH, there is no algorithm with running time $2^{o(n)}$ to compute a maximum independent set in graphs of maximum degree $3$.\end{theorem}
	
	Similarly, Fleischner et al.~\cite{FLEISCHNER20102742} showed that MIS is NP-hard even in $3$ and $4$-regular Hamiltonian planar graphs\footnote{They used a claim in book~\cite{GareyJ79} which states that MIS is NP-complete in cubic planar graphs. In the book, the authors cited the paper of Garey et al.~\cite{GareyJS76}. To our understanding, this is an incorrect referencing. However later, Mohar~\cite{Mohar01} showed the hardness of the problem in the claimed class. Therefore the result of Fleischner et al.~\cite{FLEISCHNER20102742} is valid.}. Their reduction is a bit more involved than the two others as they had to support the Hamiltonicity of the graph. All of the above constructions are specialized for their specific purpose and we do not see a direct extension of the mentioned papers to general $d$-regular graphs for every $d>2$. 
	
	Similar to the predecessor work, we also discuss the hardness of the MIS problem in regular planar graphs. We show that the simple construction for general graphs extends to $5$-regular planar graphs. Since there is no $6$-regular planar graph, this together with previous results show that in any $d$-regular planar graph ($d=3,4,5$), the problem is NP-hard.\footnote{ $d=1$ is a matching and $d=2$ is a disjoint union of cycles, thus there is a linear-time algorithm for the problem in these cases.}
	
	The closely related problems of finding and counting cliques and finding a minimum vertex cover in regular graphs will be discussed at the end. Our simple gadget construction facilitates further customization for obtaining lowerbounds on the computational complexity of covering problems. 
	
	Before we delve into the technical parts, let us introduce the notation that is used here. We denote by $\N$ the set of natural numbers and for a
	set of integers $\{1,\ldots,k\}$ we write $[k]$. The degree of a vertex $v$ in a graph $G=(V,E)$ is denoted by $d_v$. $\Delta$ stands for the maximum degree of a graph. For a vertex $v$, the closed neighborhood of $v$ is written as $N[v]$ ($N[v]$ contains $v$ and all of its neighbors). 
	\section{MIS has no Subexponential Algorithm on $d$-regular Graphs}

	We may always assume that the maximum degree $\Delta$ of the input graph is odd. 
	Otherwise, we add a complete graph on $\Delta + 2$ vertices to the original graph. The resulting graph has an odd maximum degree $\Delta + 1$ and it is clear that in polynomial time we can construct MIS of the original graph from MIS of the new graph and vice versa. Hence, we have the following assumption for the rest of the paper.
	
	\medskip
	
	\textbf{Assumption:\label{assumption} }$\Delta$ is an odd number bigger than $1$.

	\subsection*{Gadget Construction}

	%
	%
	%
	%

	For a vertex $v$ of degree $d_v < \Delta$ we construct $\delta_v =\Delta - d_v$ distinct gadgets $H^v_1,\ldots,H^v_{\delta_v}$ as follows (all of them have the same structure). In the following we explain the construction of a single gadget, let say $H$.
	
	First create $(\Delta-1)/2$ complete bipartite graphs $K_1,\ldots,K_{(\Delta-1)/2}$ with partitions of size $\Delta -1$. We name the partitions of the $i$'th bipartite graph $A_i,B_i$, for $i\in[(\Delta-1)/2]$. Add $\Delta-1$ vertices $a_1,\ldots,a_{(\Delta -1)/2},b_1,\ldots,b_{(\Delta-1)/2}$ to the gadget $H$. Connect all vertices of partition $A_i$ (resp.\ $B_i$) to $a_i$ (resp.\ $b_i$). Then connect all $a_i,b_i$'s ( $i\in[(\Delta - 1)/2]$) to a new vertex $h$. The construction of $H$ is completed. By construction, every vertex except $h$, has degree $\Delta$. The degree of $h$ is $2(\Delta - 1)/2 = \Delta - 1$. $h$ is the vertex that connects our gadget to the graph $G$.
	
	Whenever it is necessary, if a gadget $H$ is the $j$'th gadget of a vertex $v$, to distinguish different gadgets, we add indices $v,j$ to $H$ and all of the aforementioned vertices and partitions. E.g.\ instead of a vertex $h$ we may write $h^v_j$. 
	
	The construction of the auxiliary graph $G'$ is pretty simple: take $G$ as a base, then for every $v\in V(G)$ connect all of its gadgets, i.e.\ $H^v_j$'s, to $v$ by adding edges $\{h^v_j,v\}$ for $j\in [\delta_v]$. Let us make some observation on $G'$. First observe that every vertex of $G'$ has degree exactly $\Delta$.
	
	We formalize the second observation for bounding the order of $G'$ in the following.
	\begin{observation}\label{obs:size}
		The order of an attached gadget to any vertex is $O(\Delta^2)$. Since there are at most $\Delta$ such gadgets attached to a vertex $v$, $G'$ has $O(\Delta^3 |V(G)|)$ vertices. As the number of edges of each gadget is at most $\Delta$ times more than its vertices, $G'$ has $O(\Delta^4 |V(G)| + |E(G)|)$ edges.
	\end{observation}

	\subsection{From an MIS in $G'$ to an MIS in $G$}
	The main observation on each individual gadget is the following (we ignore the indices of the gadget for simplicity). In any MIS of $G'$, for a gadget $H$, from each bipartite graph $K_i$ in $H$ we have to take one of its partitions, $A_i$ or $B_i$, entirely into the MIS. The design of $H$ is such that, after the previous selection we can take either of the sets $a^i$'s or $b^i$'s in the solution. But then we are not able to take the vertex $h$ in the MIS. Consequently, vertex $v$ (a vertex of $G$ that is connected to the gadget $H$ in $G'$) is freely available to join MIS later. Hence, the existence of $v$ in MIS merely depends on the structure of $G$, not its connected gadgets. We prove these claims formally in the following.

	First we explain how to construct an MIS in a single gadget $H$.
	\begin{lemma}\label{lem:ISoddGadget}
		Let $H$ be a gadget. $H$ has a maximum independent set $I$ of size
		$(\Delta-1)^2/2 + \Delta - 1$ such that $h\not\in I$.
	\end{lemma}
	\begin{proof}
		We first constructively show that an independent set of the claimed size and structure exists; then we prove it is a maximum independent set.
		To construct $I$, take all vertices in partitions $A_i$ ($i\in [(\Delta - 1)/2]$) into $I$, then add all vertices with labels $b_i$ to $I$.
		The size of $I$ is as claimed, it does not contain a vertex $h$, and it is an independent set of $H$. It is left to show that there is no independent set $I'$ of larger size in $H$. 
		
		Clearly, we can take at most $\Delta - 1$ vertices of the $i$'th bipartite graph of the gadget in the MIS. We show that exactly $\Delta - 1$ vertices of such a bipartite graph appears in any MIS.
		
		For the sake of contradiction, suppose that in one of these $K_i$'s, let call it $K$, an MIS $I'$ of $H$ has at most $t\le \Delta - 2$ vertices of $K$. If $t>0$, then w.l.o.g.\ suppose the selected vertices of $K$ are in its $B$ part\footnote{Clearly if a vertex from the $B$ part of $K$ is in an independent set of $H$ then no vertex from its $A$ part can contribute to that independent set, as $K$ is a complete bipartite graph.}. But then if $u\in B\cap I$ then every other $v\in B$ can safely join $I$ since $N(u)=N(v)$ and there is no edge between $u$ and $v$, a contradiction to the assumption that $I$ was of the maximum size.

		It remains to show the claim holds for the case of $t=0$. $t=0$ means no vertex of $K'$ is in $I'$, then we should have both $a_i, b_i\in I$ (otherwise we add one side of $K$ to $I'$ and make a larger independent set). If this is the case, we remove $a_i$ from $I$ and add all vertices of the $A$ partition of $K$ to the independent set to make it larger, a contradiction. 
		
		Therefore, in any maximum independent set $I'$, for every bipartite graph $K_i$  one of its partitions is entirely in $I'$. For the remaining undecided vertices, observe that we may take at most $(\Delta - 1)/2$ other vertices in the maximum independent set, this is forced by the choice of the corresponding partitions of bipartite graphs.
	\end{proof}
	
	Now we are ready to establish a connection between MIS of $G$ and $G'$ by the following lemma.
	
	\begin{lemma}\label{lem:isodd}
		Given an integer $k$, there is an independent set $I'$ of $G'$ of size at least $k + \Sigma_{v\in V(G)} (\Delta - d_v) \cdot ((\Delta-1)^2/2 + \Delta - 1)$ if and only if there is an independent set $I$ of size at least $k$ in $G$. Moreover, we can construct $I$ from $I'$ and vice versa in linear time.
	\end{lemma}
	\begin{proof}
		The only if direction is straightforward: initialize $I'=I$ then add all maximum independent sets of all gadgets, computed by the approach explained in the proof of~\Cref{lem:ISoddGadget}, to $I'$. The size of $I'$ is as claimed. On the other hand, none of the vertices of gadgets that are connected to the vertices of $G$ are in $I'$. It means that there is no conflict between choices in gadgets and vertices in $I$, hence $I'$ is an independent set of the claimed size.
		
		For the if part, by~\Cref{lem:ISoddGadget} there are at most $\Sigma_{v\in V(G)} (\Delta - d_v) \cdot (\Delta-1)^2/2 + \Delta - 1)$ vertices in $I'$ that are in $G'-V(G)$. Hence, at least $k$ vertices $I=\{u_1,\ldots,u_k\}$ of $I'$ belong to both $G$ and $G'$, thus $I$ is an independent set of size $k$ in $G$.
	\end{proof}
	
	The main theorem is the consequence of the previous lemmas and the sparsification lemma for the MIS problem.
	\begin{theorem}\label{thm:main}
		There is no algorithm with running time $2^{o(|E|)}$ to solve the maximum independent set problem in $d$-regular graphs unless ETH fails.
	\end{theorem}
	\begin{proof}
		
		Let $G_1$ be the graph constructed in the lowerbound proof~\Cref{thm:ref8} (recall that $G_1$ has maximum degree $3$). For any $d>2$ define a graph $G$ to be disjoint union of $G_1$ and a star on $d+1$ vertices. Clearly, there is a subexponential algorithm to find an MIS in $G_1$ if and only if there is such an algorithm for $G$.

		As explained earlier, in the description for Assumption~\ref{assumption}, w.l.o.g.\ we may assume that $G$ has an odd maximum degree. Thus, the reduction in Lemma~\ref{lem:isodd} applied on $G$ shows the hardness of the MIS on $d$-regular graphs. By Observation~\ref{obs:size} the size of each gadget is $\poly(d)$ (independent of the order of $G$), hence the theorem follows.
	\end{proof}
	
	\subsection{Extensions}
	Our construction simply extends to vertex cover and clique problem. On the other hand, another extension is to set up a similar lower bound in planar graphs. Our gadgets are not planar but it is easy to modify the most interior part of the gadgets (the bipartite graphs) to obtain planar gadgets. We explain the case of $5$-regular planar graphs then we talk about the extension to the maximum clique problem.

	\subsection*{Regular Planar Graphs }
	
	As discussed in the introduction, it is well known that the MIS problem is hard in $3,4$-regular planar graphs. We do not know if there is any result to show the hardness for $5$-regular planar graphs. Here we present a simple construction to show the hardness of MIS (and consequently minimum vertex cover) in these graphs. The construction is similar as before, we keep vertices $a_i,b_i$ as we had, however, instead of bipartite graphs in the gadget, we insert a modified icosahedron as drawn in~\Cref{fig:icosahedron}, we call this graph $\mathcal{X}$.

	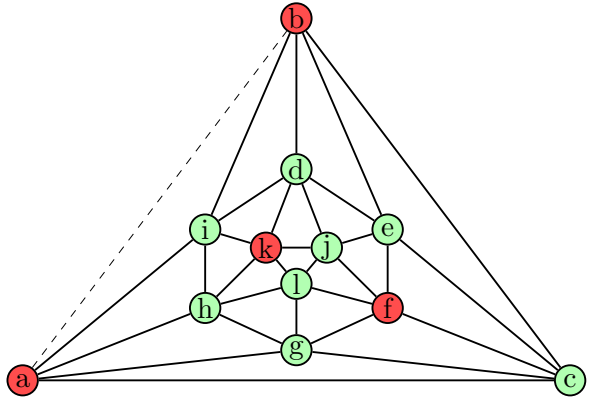
\begin{figure}
		\centering
		\begin{tikzpicture}[scale=0.8,every node/.style={draw=black,text=black,circle,fill=green!30,inner sep=0pt, minimum size=0.4cm},every path/.style={line width=0.25mm,draw=black}]
		\node[fill=red!70] (a) at (0,0) {a};
		\node[fill=red!70] (b) at (4.5,6) {b};
		\node (c) at (9,0) {c};
		
		\node (d) at (4.5,3.5) {d};
		\node (e) at (6,2.5) {e};
		\node[fill=red!70] (f) at (6,1.2) {f};
		\node (g) at (4.5,0.5) {g};
		\node (h) at (3,1.2) {h};
		\node (i) at (3,2.5) {i};
		\node (j) at (5,2.2) {j};
		\node[fill=red!70] (k) at (4,2.2) {k};
		\node (l) at (4.5,1.6) {l};

		\draw [dashed, line width=0.1mm] (a) to (b);
		\draw (a) to (c);
		\draw (b) to (c);
		\draw (d) to (e);
		\draw (e) to (f);
		\draw (f) to (g);
		\draw (g) to (h);
		\draw (h) to (i);
		\draw (i) to (d);
		\draw (b) to (d);
		\draw (b) to (i);
		\draw (b) to (e);
		\draw (i) to (a);
		\draw (h) to (a);
		\draw (g) to (a);
		\draw (e) to (c);
		\draw (f) to (c);
		\draw (g) to (c);
		\draw (l) to (k);
		\draw (j) to (k);
		\draw (j) to (l);
		\draw (j) to (e);
		\draw (j) to (f);
		\draw (j) to (d);
		\draw (l) to (f);
		\draw (l) to (g);
		\draw (l) to (h);
		\draw (k) to (h);
		\draw (k) to (i);
		\draw (k) to (d);
		\end{tikzpicture}
		\caption{Graph $\mathcal{X}$ is obtained by deleting an edge $\{a,b\}$ from an icosahedron. The vertices $a,b,k,f$ (in red) form an independent set of size $4$ and every other independent set has size less than $4$.}
		\label{fig:icosahedron}
	\end{figure}
	
	\begin{lemma}\label{lem:planar}
		$\mathcal{X}$ has a maximum independent set of size~$4$ and both vertices $a,b$ will be in any MIS.
	\end{lemma}
	\begin{proof}
		One can observe that vertices $a,b,f,k$ all together form an independent set of the claimed attributes. We prove that this is the only MIS of $\mathcal{X}$ by showing that in any MIS, both vertices $a,b$ are present. 
		
		Before proving the above claim let us explain the general idea: if $a,b$ are in an MIS $I$, then the only remaining vertices in $X'=\mathcal{X}-(N[a]\cup N[b])$ are $k,l,j,f$. In this case, it is easy to see that $I$ is actually $\{a,b,k,f\}$, since $j,l$ have degree $3$ in $X'$. The main issue arises when at least one of the vertices $a$ or $b$ are not in $I$. Here we calculate the size of neighborhood of any independent set $I'$ on $3$ vertices $x,y,u\in \mathcal{X} - \{a,b\}$. 
		
		We show that $|N[I']|=|N[x]\cup N[y]\cup N[u]| = 12$. Note that $12$ is the total number of vertices of $\mathcal{X}$. It means that $I'$ already is a maximal independent set and we cannot add any other vertex to it. Thus if $I$ is a maximum independent set, then $I$ does contain both $a,b$ and the lemma follows. Now we prove the claim that the size of neighborhood of any such triple (as an independent set) is actually 12.
		
		Except $a,b$, every other vertex has degree $5$ and every two non-adjacent vertices share at most $2$ neighbors. Hence, if there are two vertices $x,y\in V(\mathcal{X}) - \{a,b\}$ in an MIS $I$, then $|N[x]\cup N[y]\}|\ge 6+6-2=10$. The latter means that all of the vertices of $\mathcal{X}$ except at most two of them, let call them $u,v$, are in the closed neighborhood of $x,y$. Clearly both $u,v$ are in $I$ otherwise the size of $I$ is less than $4$. If $\{u,v\}=\{a,b\}$ as explained earlier, we are done. Hence, w.l.o.g.\ let suppose $u\notin \{a,b\}$. Since $|N[u]|=6$ and $u$ is not neighbor of $x,y$, we conclude that $u$ is neighbor to at least $6-2-2$ vertices that are not in $N[x]\cup N[y]$. It means that $x,y,u$ together are neighbor of all vertices of $\mathcal{X}$, hence, $v$ cannot be in $I$, a contradiction to the assumption that $I$ was an MIS. 
	\end{proof}
	
	The rest of the proof is straightforward from the above lemma and our general construction. Construct a gadget $H$ by taking $2$ copies $X_1,X_2$ of $\mathcal{X}$ and adding a vertex $h$. Then connect $a_i,b_i\in X_i$ to $h$ (we added indices to vertices of $\mathcal{X}$ to distinguish the disjoint copies of it). Eventually, attach the copies of the gadget $H$ to every vertex that has a degree less than $5$ in a given planar graph in the same way as for general graphs to obtain a $5$ regular planar graph $G'$.
	
	\begin{theorem}
		The MIS problem is NP-hard in $5$-regular planar graphs.	
	\end{theorem}
	\begin{proof}
		By the result of Mohar~\cite{Mohar01} we know that the MIS problem is NP-hard in cubic planar graphs.
		Given a cubic planar graph $G$, construct a $5$-regular planar graph $G'$ as explained above. $G$ has an independent set of size $k$ if and only if $G'$ has an independent set of size $k+4\Sigma_{v\in V(G)}(5-d_v)$. Hence, the theorem follows.
	\end{proof}
	\subsection*{Triangles and Cliques }
	The gadgets do not have a triangle as a subgraph, on the other hand, the original connections in the graph $G$ are untouched, hence there is a clique on at least ~$k\ge 3$ vertices in $G'$ if and only if there is a clique of order $k$ in $G$. Since the transformation from $G$ to $G'$ happens in linear time on graphs of bounded degree, essentially every hardness result, in graphs of bounded degree, for finding triangles or small cliques extends to the regular graphs.
	
	\section{Conclusion and Future Directions}
	In this work, we showed that the maximum independent set problem has no subexponential algorithm in $d$-regular graphs.
	Our construction, with simple modifications, extends to other covering problems and also to other classes of graphs. We believe this work could ease the way to obtain fine-grained reductions for other problems.
	
	We considered the independent set problem, one of the most basic problems were its sparsification lemma is known. Another interesting direction is to consider the $k$-SAT problem when the corresponding graph has the same degree for all variables and clauses.
	
	\medskip 
	
	\textbf{Acknowledgement: } We thank Kevin Schewior, Sebastian Siebertz, James Preen, Hossein Vahidi and, anonymous reviewers for their improvement suggestions.
	\clearpage
	\bibliographystyle{abbrv}
	\bibliography{references}

\begin{thebibliography}{10}

\bibitem{BacsoLMPTL19}
G.~Bacs{\'{o}}, D.~Lokshtanov, D.~Marx, M.~Pilipczuk, Z.~Tuza, and E.~J. van
  Leeuwen.
\newblock Subexponential-time algorithms for maximum independent set in
  {\textdollar}{\textdollar}p{\_}t{\textdollar}{\textdollar} {P} t -free and
  broom-free graphs.
\newblock {\em Algorithmica}, 81(2):421--438, 2019.

\bibitem{BourgeoisEPR12}
N.~Bourgeois, B.~Escoffier, V.~T. Paschos, and J.~M.~M. van Rooij.
\newblock Fast algorithms for max independent set.
\newblock {\em Algorithmica}, 62(1-2):382--415, 2012.

\bibitem{Brause17}
C.~Brause.
\newblock A subexponential-time algorithm for the maximum independent set
  problem in pt-free graphs.
\newblock {\em Discret. Appl. Math.}, 231:113--118, 2017.

\bibitem{DemaineFHT04}
E.~D. Demaine, F.~V. Fomin, M.~T. Hajiaghayi, and D.~M. Thilikos.
\newblock Subexponential parameterized algorithms on graphs of bounded-genus
  and \emph{H}-minor-free graphs.
\newblock In J.~I. Munro, editor, {\em Proceedings of the Fifteenth Annual
  {ACM-SIAM} Symposium on Discrete Algorithms, {SODA} 2004, New Orleans,
  Louisiana, USA, January 11-14, 2004}, pages 830--839. {SIAM}, 2004.

\bibitem{FLEISCHNER20102742}
H.~Fleischner, G.~Sabidussi, and V.~I. Sarvanov.
\newblock Maximum independent sets in 3- and 4-regular hamiltonian graphs.
\newblock {\em Discrete Mathematics}, 310(20):2742 -- 2749, 2010.
\newblock Graph Theory — Dedicated to Carsten Thomassen on his 60th Birthday.

\bibitem{FominGK09}
F.~V. Fomin, F.~Grandoni, and D.~Kratsch.
\newblock A measure {\&} conquer approach for the analysis of exact algorithms.
\newblock {\em J. {ACM}}, 56(5):25:1--25:32, 2009.

\bibitem{GareyJ79}
M.~R. Garey and D.~S. Johnson.
\newblock {\em Computers and Intractability: {A} Guide to the Theory of
  NP-Completeness}.
\newblock W. H. Freeman, 1979.

\bibitem{GareyJS76}
M.~R. Garey, D.~S. Johnson, and L.~J. Stockmeyer.
\newblock Some simplified np-complete graph problems.
\newblock {\em Theor. Comput. Sci.}, 1(3):237--267, 1976.

\bibitem{ImpagliazzoPZ01}
R.~Impagliazzo, R.~Paturi, and F.~Zane.
\newblock Which problems have strongly exponential complexity?
\newblock {\em J. Comput. Syst. Sci.}, 63(4):512--530, 2001.

\bibitem{JohnsonS99}
D.~S. Johnson and M.~Szegedy.
\newblock What are the least tractable instances of max tndependent set?
\newblock In R.~E. Tarjan and T.~J. Warnow, editors, {\em Proceedings of the
  Tenth Annual {ACM-SIAM} Symposium on Discrete Algorithms, 17-19 January 1999,
  Baltimore, Maryland, {USA}}, pages 927--928. {ACM/SIAM}, 1999.

\bibitem{Mohar01}
B.~Mohar.
\newblock Face covers and the genus problem for apex graphs.
\newblock {\em J. Comb. Theory, Ser. {B}}, 82(1):102--117, 2001.

\bibitem{PaturiPSZ05}
R.~Paturi, P.~Pudl{\'{a}}k, M.~E. Saks, and F.~Zane.
\newblock An improved exponential-time algorithm for \emph{k}-sat.
\newblock {\em J. {ACM}}, 52(3):337--364, 2005.

\bibitem{Robson86}
J.~M. Robson.
\newblock Algorithms for maximum independent sets.
\newblock {\em J. Algorithms}, 7(3):425--440, 1986.

\bibitem{XiaoN17}
M.~Xiao and H.~Nagamochi.
\newblock Exact algorithms for maximum independent set.
\newblock {\em Inf. Comput.}, 255:126--146, 2017.

\bibitem{clique}
D.~Zuckerman.
\newblock Linear degree extractors and the inapproximability of max clique and
  chromatic number.
\newblock In {\em STOC'06}, volume 2006, pages 681--690, 9 2006.

\end{thebibliography}
		
\end{document}